\documentclass[pra,10pt,twocolumn,superscriptaddress,floatfix,noshowpacs]{revtex4}

\usepackage[scaled=0.86]{berasans}  
\usepackage[colorlinks=true, allcolors=blue, urlcolor=blue]{hyperref}  
\usepackage{graphicx} 
\usepackage[babel]{microtype}  
\usepackage{amsmath,amssymb,amsthm,bm,amsfonts,mathrsfs,bbm} 
\usepackage{pgfplots}
\usepackage{array}
\usepackage{bigstrut}
\usepackage{algorithm} 
\usepackage{algpseudocode} 
\algnewcommand{\To}{\textbf{to }}
\algnewcommand\Input{\item[\textbf{input:}]}%
\algnewcommand\Output{\item[\textbf{output:}]}%

\algdef{SE}[DOWHILE]{Do}{doWhile}{\algorithmicdo}[1]{\algorithmicwhile\ #1}%

\newcommand{\tr}{\text{tr}}
\newcommand{\ket}[1]{| #1 \rangle}

\newcommand{\be}{\begin{equation}}
\newcommand{\ee}{\end{equation}}
\newcommand{\bea}{\begin{eqnarray}}
\newcommand{\eea}{\end{eqnarray}}
\newcommand{\bes}{\begin{equation*}}
\newcommand{\ees}{\end{equation*}}
\newcommand{\beas}{\begin{eqnarray*}}
	\newcommand{\eeas}{\end{eqnarray*}}

\def\v{\vec{v}}

\def\tr{\mathrm{tr}}

\newtheorem{thm}{Theorem}
\newtheorem*{thm*}{Theorem}

\newtheorem*{lem*}{Lemma}

\newtheorem{conjecture}{Conjecture}

\newtheorem*{lipschitzLem*}{Lemma \ref{lipschitz}}
\newtheorem*{lipschitzCubeLem*}{Lemma \ref{lipschitzCube}}
\newtheorem*{pgmNearlyOptimalThm*}{Theorem \ref{pgmNearlyOptimal}}

\usepackage[cp1250]{inputenc}
\usepackage{amsmath}
\usepackage{amsfonts}
\usepackage{graphicx}
\usepackage{yfonts}
\usepackage{wasysym}
\usepackage{color}
\usepackage[normalem]{ulem}
\usepackage{amsthm}
\usepackage{bm}
\usepackage{bbm}
\usepackage{mathtools}

\def\<{\langle}
\def\>{\rangle}
\newtheorem{Prop}{Proposition}
\usepackage{placeins}
\def\oper{{\mathchoice{\rm 1\mskip-4mu l}{\rm 1\mskip-4mu l}
{\rm 1\mskip-4.5mu l}{\rm 1\mskip-5mu l}}}

\begin{document}

\title{ How many mutually unbiased bases are needed to detect bound entangled states?}

\author{Joonwoo Bae}
\affiliation{
School of Electrical Engineering, Korea Advanced Institute of Science and Technology (KAIST), 291 Daehak-ro Yuseong-gu, Daejeon 34141 Republic of Korea,}

\author{Anindita Bera}
\affiliation{ Institute of Physics, Faculty of Physics, Astronomy and Informatics  Nicolaus Copernicus University, Grudzi\c{a}dzka 5/7, 87--100 Toru\'n, Poland,}

\author{Dariusz Chru\'sci\'nski }
\affiliation{ Institute of Physics, Faculty of Physics, Astronomy and Informatics  Nicolaus Copernicus University, Grudzi\c{a}dzka 5/7, 87--100 Toru\'n, Poland,}

\author{Beatrix C. Hiesmayr}
\affiliation{University of Vienna, Faculty of Physics, W\"ahringer Stra{\ss}e 17, 1090 Vienna, Austria,}

\author{Daniel McNulty}
\affiliation{Center for Theoretical Physics, Polish Academy of Sciences, Al. Lotnik\'ow 32/46, 02-668 Warsaw, Poland,}
\affiliation{Department of Mathematics,  Aberystwyth University, Aberystwyth, Wales, U.K.}

\begin{abstract}
From a practical perspective it is advantageous to develop methods that verify entanglement in quantum states with as few measurements as possible. In this paper we investigate the minimal number of mutually unbiased bases (MUBs) needed to detect bound entanglement in bipartite $(d\times d)$-dimensional states, i.e. entangled states that are positive under partial transposition. In particular, we show that a class of entanglement witnesses composed of mutually unbiased bases can detect bound entanglement if the number of measurements is greater than $d/2+1$. This is a substantial improvement over other detection methods, requiring significantly fewer resources than either full quantum state tomography or measuring a complete set of $d+1$ MUBs. Our approach is based on a partial characterisation of the (non-)decomposability of entanglement witnesses. We show that non-decomposability is a universal property of MUBs, which holds regardless of the choice of complementary observables, and we find that both the number of measurements and the structure of the witness play an important role in the detection of bound entanglement.


\end{abstract}

\maketitle

\section{Introduction}

The extensive study of correlations in quantum systems has revealed a complex hierarchy of entanglement, each class with its own distinct non-classical properties and applications. Examples include entangled states which are unsteerable (admitting a local hidden state model \cite{wiseman07}), steerable states (which admit only a local hidden variable model \cite{schro35}), and entangled states which exhibit Bell non-locality \cite{bell64}. There is also a class of states, known as bound entangled states, from which no pure entangled state can be distilled \cite{PhysRevLett.80.5239,PhysRevLett.82.1056}.  Any entangled state that remains positive under partial transposition (PPT) is bound entangled, but it is not yet known if all bound entangled states are PPT.

Entanglement is a fundamental ingredient in almost all quantum information tasks. Usually, more entanglement in a system implies a higher degree of performance for certain tasks~\cite{PhysRevLett.122.140402, PhysRevLett.122.140403, PhysRevLett.122.140404} and, consequently, highly entangled states are a valuable resource for quantum technologies. While bound entangled states are considered weakly entangled, they can be used for quantum steering (i.e. ruling out local hidden state models \cite{moroder14}), the distillation of secure quantum keys \cite{horodecki05,horodecki08}, and for quantum teleportation~\cite{PhysRevLett.96.150501}. They also provide valuable insights into the study of classical information theory, providing important connections to bound information \cite{gisin00}.

Certifying non-classical correlations is both of fundamental and practical interest \cite{RevModPhys.81.865, GUHNE20091, Chru_ci_ski_2014}, and has been observed at low and high energies (e.g. in the entangled K-meson system~\cite{KLOE}, photons beyond the optical wave length~\cite{JPET}, or oscillating neutrino systems~\cite{HieNeutrino}).

 Given the full information about a quantum system, it is, in general, computationally hard to reveal entanglement~\cite{10.1145/780542.780545}. More precisely, the difficulty lies in the characterisation of those entangled states that remain positive under partial transposition. So far, only a few families of bound entangled states have been constructed~\cite{Darek,Sanperabound,horodecki97,Bennett,divincenzo03,Krammer,Lokard,Bruss,HiesmayrLoeffler2,tura18}, each by very different tools, and a general characterisation of bound entangled states is missing.

There are several criteria which enable the detection of entangled states (cf. \cite{RevModPhys.81.865,GUHNE20091,Chru_ci_ski_2014}). A popular detection scheme is based on the concept of an entanglement witness (EW), a Hermitian operator $\mathbf{W}$ that satisfies
\bea
\tr[\mathbf{W}\rho_{\mathrm{sep}}] &\geq& 0~~\mathrm{for~all~separable~states}~\rho_{\mathrm{sep}}, \nonumber \\
\mathrm{and}~\tr[\mathbf{W}\rho] &<& 0 ~~\mathrm{for~some~entangled~ states}~\rho. \label{eq:def}
\eea
Hierarchical structures on the level of certifications of entanglement have been recently introduced according to the assumptions made in concrete experimental setups~\cite{Supic2020selftestingof, PRXQuantum.2.010201}. In the \textit{standard} scenario all experimental devices are trusted, whereas another extreme scenario is a \textit{device-independent} (DI) one, which demands that the certification does not rely on the details of the source and the measurement devices~\cite{PhysRevLett.121.180503}.  An intermediate scenario is \textit{measurement-device-independent} (MDI) certification that is supposed to close all loopholes that may appear in the detection devices. In this case, the certification can be straightforwardly constructed from EWs~\cite{PhysRevLett.110.060405, PhysRevLett.116.190501}, which we pursue in this contribution.

Note that EWs correspond to a particular set of observables that satisfy Eq. (\ref{eq:def}). They can be decomposed into local measurements, where it is of practical importance to minimise the number of outcomes, or precisely, positive-operator-valued-measure (POVM) elements. An interesting class of EWs arises from sets of mutually unbiased bases (MUBs), that are experimentally the most feasible and also gradually approach quantum state tomography as they converge to being complete.

A pair of orthonormal bases $\mathcal{B}_1$ and $\mathcal{B}_2$ of the space $\mathbb{C}^d$ are mutually unbiased if the basis vectors satisfy $|\langle \psi_i | \varphi_j \rangle |^2 = \frac{1}{d}$ for all $\ket{\psi_i}\in\mathcal{B}_1$ and $\ket{\varphi_j}\in\mathcal{B}_2$. This condition mathematically formalises the concept of complementarity, i.e., if the system resides in an eigenstate of one observable, the outcome of measuring a second complementary observable is maximally uncertain. Generalising to larger collections of measurements, we say a set of bases is mutually unbiased if all bases are pairwise unbiased. As is well known (see \cite{durt10} for a review), complete sets of $d+1$ MUBs have been constructed in prime and prime-power dimensions \cite{ivanovic,wootters}, but it is conjectured that only smaller sets exist in all other cases \cite{boykin07}.

The first direct connection between MUBs and EWs was made in \cite{PhysRevA.86.022311}, with the construction of witnesses composed of two or more MUBs. Generalisations of the witnesses to a wider class of measurements, including 2-designs and mutually unbiased measurements (MUMs), for which MUBs are an example, have since appeared \cite{darek18,kdesign,hiesmayr20,bae13,erker17,kaniewski17,kalev17,bavaresco18}. In low dimensions, equipped with a complete set of MUBs, examples of these witnesses have been shown to experimentally detect bound entangled states \cite{Hiesmayr_2013}. However, little is known about their capability to detect PPT entangled states in arbitrary dimensions and the number of measurements required.

The ability of an EW to detect bound entangled states is directly connected to the notion of non-decomposability \cite{PhysRevLett.82.1056,tehral01,lewenstein00,lewenstein01,ha12}. An entanglement witness $\mathbf{W}$ is decomposable if $\mathbf{W}=A+ B^\Gamma$, with $A,B \geq 0$ ($\Gamma$ denotes the partial transposition). EWs that cannot be decomposed in this way (i.e. non-decomposable) are able to detect bound entangled states. In fact, an EW is non-decomposable if and only if it detects PPT entangled states \cite{lewenstein00}. Unfortunately, there is no general method to construct such objects, and it is usually difficult to decide whether a witness is decomposable or not. One focus of this paper will be to partially ameliorate this difficulty by constructing a general class of non-decomposable witnesses from small sets of MUBs.

Another popular entanglement detection approach implements tomographically complete measurements to recover all information about the given quantum state. With full knowledge of the state, any known detection tool can then be applied to certify entanglement. In this scenario, complete sets of MUBs also appear, as they provide an optimal measurement scheme for quantum state tomography \cite{ivanovic,wootters}. For example, bound entanglement detection via quantum state tomography has been experimentally implemented several times before in the multipartite setting \cite{amselem09,lavole10,barreiro10,kampermann10,kaneda12}, and less frequently in the bipartite scenario \cite{diguglielmo11,Hiesmayr_2013}.

The two certification methods described above and their connection with MUBs is particularly convenient due to the relatively simple implementation of these observables experimentally. For example, experimental implementations of MUBs have been demonstrated in a wide range of quantum information protocols, e.g. \cite{Hiesmayr_2013,ecker19,herrera20,giovannini13,dambrosio13,bartkiewicz16,mafu13,adamson10,lukens18,hu20}, as well as in fairly large dimensional systems \cite{kewming19,bouchard18}. The EW approach has the added bonus of requiring at most $d+1$ MUBs, measured on each subsystem, while full quantum state reconstruction requires $d^2+1$ MUBs in $\mathbb{C}^d\otimes\mathbb{C}^d$ (or any other informationally complete set of measurements on the full state space), and often suffers from errors, such as the stability of the source over all possible measurements  \cite{schwemmer15,sentis18}. In general dimensions (not restricted to prime powers) one can use weighted 2-designs consisting of orthonormal bases to perform quantum state tomography but this requires more bases than MUBs \cite{roy07}. Known witnesses for bound entanglement, e.g., the computable cross-norm or realignment (CCNR) criterion \cite{rudolph03}, similarly require of the order $d^2$ measurements. Methods with as few measurements as possible are therefore highly desirable.

In this work, our primary motivation is to ask whether a complete set of MUBs in $\mathbb{C}^d$ is needed to detect bound entangled states. First, we consider which of the previously constructed EWs, composed of a complete set of MUBs, are useful to detect bound entanglement. We show (Sec. \ref{sec:full_decomposable}) that the original witness described in \cite{PhysRevA.86.022311} is decomposable for any $d$, and fails to detect PPT entangled states. We then demonstrate (Sec. \ref{sec:full_nondecomposable}) that modifying the EW (via re-orderings of a single basis in one subsystem \cite{Hiesmayr_2013,darek18}), results in a non-decomposable witness that detects PPT entangled states in any dimension $d$, provided a complete set of MUBs exists. We derive our main result in Sec. \ref{sectionlessMUBs}, which shows that bound entanglement is detected by \emph{any} set of $m$ MUBs (in every dimension $d$) when $m>d/2+1$. This requires significantly fewer measurements than previous approaches. In Sec. \ref{sec:2mubs}, we show that a pair of Fourier connected MUBs is incapable of certifying bound entanglement. We then conclude with examples in $d=3$ (Sec. \ref{sec:d=3}) and $d=2^r$ (Sec. \ref{sec:special}). In particular, we show that three MUBs are both necessary and sufficient to detect bound entanglement when $d=3$. We conjecture that bound entanglement can be verified with only $m=d/2+1$ MUBs when $d=2^{r}$, $r\geq 2$. This is corroborated with examples, in $d=4$, of non-decomposable witnesses using only three measurements, and as a by-product we present new families of bound entangled states that they detect.

Overall, our work illustrates that MUBs are a powerful tool in the study of quantum correlations, useful not only to distinguish between separable and entangled states, but also to penetrate deeper into the full hierarchy of non-classical correlations. This contribution, together with recent work connecting MUBs with Bell nonlocality \cite{tavakoli21} and quantum steering \cite{zhu16}, is further evidence of the fundamental and practical importance of this special class of measurements in revealing non-classical correlations.

\section{Entanglement witnesses with MUBs}

Let us start with the framework of entanglement witnesses presented in \cite{Bae:2020ty}, where it is clear that EWs are built from a collection of measurements. Suppose that a non-negative operator $C$ is decomposed as
\bea
C = \sum_{a,b} A_a \otimes B_b\,, \nonumber
\eea
where $\{ A_a\}$ and $ \{B_b \}$ POVM elements and refer to the two measurable subspaces of the quantum system of interest. There are non-trivial upper and lower bounds satisfied by all separable states $\rho_{\mathrm{sep}}$, i.e.
\bea
L(C) \leq \tr[\rho_{\mathrm{sep}} C] \leq U(C)\,. \nonumber
\eea
Each of the bounds corresponds to an EW. Since upper bounds are useful to detect PPT entangled states \cite{Hiesmayr_2013,darek18}, we consider EWs that correspond to upper bounds. That is, a witness $\mathbf{W}$ is constructed such that
\bea
\tr[\mathbf{W} \rho] \ngeq 0 \iff \tr[ C \rho ] \nleq U(C)\,, \nonumber
\eea
for entangled states $\rho$ detected by an upper bound.

In what follows, we derive EWs from a collection of MUBs. Let us consider $m$ MUBs in a $d$-dimensional Hilbert space,
\bea
\mathcal{M}_m = \{ \mathcal{B}_0,\mathcal{B}_1,\ldots,\mathcal{B}_{m-1}\}\,,\nonumber
\eea
where $\mathcal{B}_\alpha =\{ |i_\alpha\> \}_{i=0}^{d-1}$, and $\mathcal{B}_0$ is chosen as the canonical basis $\{|0\>,\ldots,|d-1\>\}$. From these MUBs we construct the following operators
\bea
 \mathbf{B}(\mathcal{M}_m,s) &=& \sum_{\ell=0}^{d-1} |\ell\>\<\ell| \otimes  |\ell +s \>\<\ell+s| \nonumber  \\ &+&  \sum_{\alpha=1}^{m-1} \sum_{i=0}^{d-1} |i_\alpha \>\<i_\alpha| \otimes |i_\alpha \>\<i_\alpha|, \nonumber
\eea
and its partial transpose
\bea
 \mathbf{B}^\Gamma(\mathcal{M}_m,s) &=& \sum_{\ell=0}^{d-1} |\ell\>\<\ell| \otimes  |\ell +s \>\<\ell+s| \nonumber  \\ &+&  \sum_{\alpha=1}^{m-1} \sum_{i=0}^{d-1} |i_\alpha \>\<i_\alpha| \otimes |i_\alpha^* \>\<i_\alpha^*|, \nonumber
\eea
where the complex conjugation (and transposition) is performed with respect to the canonical basis $\mathcal{B}_0$, and the sum $\ell+s$ is defined modulo $d$. It has been proven~\cite{PhysRevA.86.022311} that for any separable state, the two operators are bounded from above by the same quantity, depending on the number of MUBs, $m$, and the dimension $d$. In particular,
\begin{equation}\label{}
\tr[ \mathbf{B}(\mathcal{M}_m,s) \, \rho_{\rm sep}] \leq \frac {d+m-1}{d} ,
\end{equation}
and
\begin{equation}\label{}
\tr[\mathbf{B}^\Gamma(\mathcal{M}_m,s)\, \rho_{\rm sep}] \leq \frac {d+m-1}{d} .
\end{equation}
Hence, the operator
\begin{equation}\label{witness}
  \mathbf{W}(\mathcal{M}_m,s) = \frac {d+m-1}{d} \oper_d \otimes \oper_d -  \mathbf{B}(\mathcal{M}_m,s) ,
\end{equation}
with $\oper_d$ the $d\times d$ identity matrix, satisfies $\tr[\mathbf{W}(\mathcal{M}_m,s)\, \rho_{\rm sep}] \geq 0$. Similarly,  $\tr[\mathbf{W}^\Gamma(\mathcal{M}_m,s) \,\rho_{\rm sep}] \geq 0$. It is, therefore, clear that whenever $\mathbf{W}(\mathcal{M}_m,s)$ possesses a strictly negative eigenvalue it defines an entanglement witness.

\section{Witnesses with $s=0$ do not detect bound entanglement}\label{sec:full_decomposable}

Consider first the simplest scenario $s=0$, such that $\mathbf{W}^{\Gamma}(\mathcal{M}_m,0)$ corresponds to the witness first considered in \cite{PhysRevA.86.022311}. Suppose that the $d$-dimensional Hilbert space $\mathcal{H}=\mathbb{C}^d$ admits a complete set of $d+1$ MUBs. The witness of Eq. (\ref{witness}) is therefore 
\begin{equation}\label{}
  \mathbf{W}(\mathcal{M}_{d+1},0) = 2 \oper_d \otimes \oper_d -  \sum_{\alpha=0}^d \sum_{i=0}^{d-1} |i_\alpha\>\<i_\alpha| \otimes |i_\alpha\>\<i_\alpha|  .
\end{equation}
We can then apply the 2-design property of a complete set of MUBs \cite{klapp05}, i.e.
\begin{equation}\label{}
 \sum_{\alpha=0}^d \sum_{i=0}^{d-1} |i_\alpha\>\<i_\alpha| \otimes |i_\alpha\>\<i_\alpha| \;=\;  2\; \Pi_{\rm sym}  ,
\end{equation}
where $\Pi_{\rm sym} = \frac 12 (\oper_d \otimes \oper_d + \mathbb{F}_d )$ denotes a projector onto the symmetric subspace of $\mathcal{H} \otimes \mathcal{H}$. Thus, we have $\mathbf{W}(\mathcal{M}_{d+1},0)=2 (\oper_d \otimes \oper_d -\Pi_{\rm sym})$. Here, $\mathbb{F}_d$ denotes a flip (swap) operator on $\mathbb{C}^d \otimes \mathbb{C}^d$ defined by
\begin{equation}\label{}
 \mathbb{F}_d  = \sum_{i,j=0}^{d-1} |i \>\<j| \otimes |j \>\<i| .
\end{equation}
Actually, $\mathbb{F}_d$ does not depend on the basis.  One has, therefore,
\begin{equation}\label{W0}
   \mathbf{W}(\mathcal{M}_{d+1},0) = 2\, \Pi_{\rm asym} ,
\end{equation}
where $\Pi_{\rm asym} = \frac 12(\oper_d \otimes \oper_d - \mathbb{F}_d)$ denotes a projector onto the antisymmetric subspace in $\mathcal{H} \otimes \mathcal{H}$. Clearly, $  \mathbf{W}(\mathcal{M}_{d+1},0) \geq 0$, and it is therefore not an entanglement witness (recall, that a witness $\mathbf{W}$ has to satisfy  $\tr[ \mathbf{W}\rho_{\rm sep}] \geq 0$ and must possess at least one strictly negative eigenvalue).   However, its partial transpose
\begin{equation}\label{}
   \mathbf{W}^\Gamma(\mathcal{M}_{d+1},0) = \oper_d \otimes \oper_d - dP^+_d ,
\end{equation}
where
\begin{equation}\label{}
  P^+_d = \frac 1d \mathbb{F}_d^\Gamma = \frac 1d \sum_{i,j=0}^{d-1} |i \>\<j| \otimes |i \>\<j| ,
\end{equation}
defines a decomposable entanglement witness corresponding to the well known reduction map \cite{cerf99,horodecki99,darek18}:
\begin{equation}\label{}
  R(X) = \oper_d \,{\rm Tr}X - X .
\end{equation}
It proves that for a trivial permutation, i.e. $s=0$, and access to a complete set of $d+1$ MUBs, one cannot detect bound entangled states. Instead, it detects, e.g., all possible entangled states of the family of isotropic states $\rho_{\rm{iso}}=\frac{1-p}{d^2} \mathbbm{1}_{d^2}+p\; P_{k\ell}$ (where $P_{k\ell}$ are the maximally entangled Bell states, defined in Eq. (\ref{Bell})).


\section{Full set of $d+1$ MUBs and $s > 0$}\label{sec:full_nondecomposable}

To construct non-decomposable witnesses one needs to consider a nontrivial shift $s > 0$ ($s$ is defined ${\rm mod} \ d$), i.e. \begin{equation}\label{}
  \mathbf{W}(\mathcal{M}_{d+1},s) = 2 \oper_d \otimes \oper_d -  \mathbf{B}(\mathcal{M}_{d+1},s).
\end{equation}
Using the 2-design property one immediately finds
\begin{eqnarray}\label{Ws}
  \mathbf{W}(\mathcal{M}_{d+1},s) &=&  \oper_d \otimes \oper_d + (\Pi_0 - \Pi_s) - \mathbb{F}_d \nonumber \\
  &=& 2 \,\Pi_{\rm asym} + (\Pi_0 - \Pi_s) ,
\end{eqnarray}
where
\begin{equation}\label{}
  \Pi_s = \sum_{i=0}^{d-1} |i\>\<i| \otimes |i+s\>\<i+s| .
\end{equation}
Note, that $\Pi_k \Pi_\ell = \delta_{k\ell} \Pi_k$, and  $\,  \Pi_0 + \ldots + \Pi_{d-1} = \oper_d \otimes \oper_d$. Interestingly, one has
\begin{eqnarray}\label{Ws0}
  \mathbf{W}(\mathcal{M}_{d+1},s) =   \mathbf{W}(\mathcal{M}_{d+1},0) + \Pi_0 - \Pi_s ,
\end{eqnarray}
or, equivalently
\begin{eqnarray}\label{Ws0s}
  \mathbf{W}(\mathcal{M}_{d+1},s) + \Pi_s =   \mathbf{W}(\mathcal{M}_{d+1},0) + \Pi_0 .
\end{eqnarray}
Now, we investigate what happens if we consider the transpose. Indeed, it turns out that $ \mathbf{W}^\Gamma(\mathcal{M}_{d+1},s)$ is a Bell diagonal bipartite operator. To show this, let us introduce a set of unitary Weyl operators
\begin{equation}\label{}
  U_{k\ell} |m\> = \omega^{k(m-\ell)}|m - \ell\>\,,
\end{equation}
where $\omega = e^{2\pi i/d}$.  Note that $U_{00}=\oper_d$,\,$\,   \tr[ U_{k\ell} U^\dagger_{mn}] = d\; \delta_{km} \delta_{\ell n}$\,, and
\begin{equation}\label{}
  U_{mn} U_{k\ell} = \omega^{nk} U_{m+k,n+\ell} \ ,\ \ \ U_{k\ell}^\dagger = \omega_d^{k \ell} U_{-k,-\ell} .
\end{equation}
We define the rank-1 projectors, corresponding to maximally entangled states,
\begin{equation}\label{Bell}
  P_{k\ell} := |\psi_{k\ell}\>\<\psi_{k\ell}| ,
\end{equation}
where
\begin{equation}\label{}
 |\psi_{k\ell}\> = \oper_d \otimes U_{k\ell} |\psi^+_d\> ,
\end{equation}
and   $|\psi^+_d\> = \frac{1}{\sqrt{d}} \sum_{n=0}^{d-1} \ket{n} \otimes \ket{n}$ denotes a canonical maximally entangled state.   A set of $d^2$ vectors
$|\psi_{k\ell}\>$---called generalised Bell vectors---defines an orthonormal basis in $\mathcal{H} \otimes \mathcal{H}$. Note that $\Pi_{\ell} := P_{0\ell} + P_{1\ell} + \ldots + P_{d-1,\ell}$, and whence
\begin{eqnarray}\label{WsG}
\mathbf{W}^\Gamma(\mathcal{M}_{d+1},s) &=& \oper_d \otimes \oper_d + (\Pi_0 - \Pi_s) - dP^+_d\nonumber\\
  &=& \oper_d \otimes \oper_d + \sum_{k=0}^{d-1} (P_{k0} - P_{ks}) - dP^+_d \nonumber\\
  &=& \mathbf{W}_{\rm Bell}(s)\,.
\end{eqnarray}

Now we show that the operator $\mathbf{W}_{\rm Bell}(s)$ defines a witness capable of detecting bound entanglement if the shift is nontrivial. In particular:
\begin{Prop}   Let $s>0$ and $d>2$.  Then  ${\mathbf{W}}_{\rm Bell}(s)$ defines a non-decomposable entanglement witness if and only if  $s \neq d/2$.
\end{Prop}
\begin{proof} Consider the following state,
\begin{eqnarray}\label{rho-x}
  \rho_x &= &\frac{1}{\mathcal{N}} ( [\oper_d \otimes \oper_d - \Pi_0 - \Pi_s - \Pi_{d-s}]  + \frac 1x  \Pi_s +
  \nonumber\\
  && x \Pi_{d-s} + d P_{00} )\,,
\end{eqnarray}
with $\mathcal{N}$ being the normalisation factor, and $x > 0$. This construction requires $s \neq d-s$, that is, $d \neq 2s$. Moreover, $s \to d-s$ corresponds to $x \to \frac 1x$. Note, that $\rho_x$ is a PPT state. Indeed, positivity of $\rho_x$ follows from 
$$ \oper_d \otimes \oper_d - \Pi_0 - \Pi_s - \Pi_{d-s} \geq 0 , $$
since $\Pi_k$ are mutually orthogonal projectors. Now, to show that $\rho_x^\Gamma \geq 0$, let us observe due to $\Pi_k^\Gamma = \Pi_k$ (these are diagonal matrices), one has
\begin{eqnarray*}\label{}
  \rho^\Gamma_x &= &\frac{1}{\mathcal{N}} ( [\oper_d \otimes \oper_d - \Pi_0 - \Pi_s - \Pi_{d-s}]  + \frac 1x  \Pi_s +
  \nonumber\\
  && x \Pi_{d-s} + d P^\Gamma_{00} )\ .
\end{eqnarray*}
Recalling, that $d P^\Gamma_{00} = \mathbb{F}_d = \sum_{k,l} |k\>\<l| \otimes |l\>\<k|$, we observe that $\rho^\Gamma_x$ consists of $2 \times 2$ blocks corresponding to the following $2\times 2$ matrices 
$$   \begin{pmatrix}
       1 & 1 \\
       1 & 1
     \end{pmatrix} \ , \ \ \ \begin{pmatrix}
       x & 1 \\
       1 & 1/x
     \end{pmatrix} . $$
Clearly, both matrices are positive definite and, hence, the positivity of $\rho_x^\Gamma$ follows. 

Straightforwardly, we obtain
\begin{eqnarray}\label{}
  \tr[ \mathbf{W}_{\rm Bell}(s) \rho_x ] &=& \frac{1}{\mathcal{N}} (d(d-2) + dx - d(d-1))\nonumber\\
   &=& \frac{d}{\mathcal{N}}(x - 1) ,
\end{eqnarray}
and, hence, $ \tr[ \mathbf{W}^\Gamma(\mathcal{M}_{d+1},s) \rho_x ] < 0$ whenever $x < 1$. Thus, $ \mathbf{W}_{\rm Bell}(s)$ detects the PPT entangled state $\rho_x$ with $x<1$ and $ \mathbf{W}_{\rm Bell}(s)$ is a non-decomposable entanglement witness. 

Now, if $s = \frac d2$, then $\mathbf{W}_{\rm Bell}(s) = A(s) + B^\Gamma(s)$, where
\begin{eqnarray}\label{}
&&A(s) = \sum_{n=0}^{s-1} (|n n \> - |(n+s) (n+s)\>)\nonumber\\
&&\qquad\qquad(\< n n| - \< (n+s) (n+s)|)  ,
\end{eqnarray}
and
\begin{equation}\label{}
  B(s) = \sum_{i< j \neq i+s} (|i j\> - |j i\>)(\< i j| - \<j i|) .
\end{equation}
Clearly, $A(s) \geq 0$ and $B(s)\geq 0$  (being sums of rank-1 projectors) which proves that $\mathbf{W}_{\rm Bell}(s)$ is decomposable.
\end{proof}

\section{$d+1$ MUBs are not necessary to detect bound entanglement}~\label{sectionlessMUBs}

Interestingly, the full set of $d+1$ MUBs is not necessary to detect bound entanglement. The main result of our paper consists of the following:
\begin{thm}\label{mainresult} Let $2s \neq d$, with $s>0$ and $d>2$. If $m > \frac d2 + 1$, then  $\mathbf{W}^{\Gamma}(\mathcal{M}_m,s)$ defines a non-decomposable entanglement witness, i.e. detects bound entangled states.
\end{thm}
\begin{proof} Let us start by writing the witness as
\begin{equation}\label{Wijkl}
  \mathbf{W}^{\Gamma}(\mathcal{M}_m,s) = \sum_{i,j,k,l=0}^{d-1} W_{ij;kl} |i\>\<j| \otimes |l\>\< k|,
\end{equation}
with $W_{ij;kl}$ the elements of a matrix $W$. It is clear that the matrix elements $W_{ij;kl}$ depend on a particular set $\mathcal{M}_m$ of MUBs. One can, however, find a subset of $\{ W_{ij;kl} \}$ which is universal, i.e., independent of the choice of MUBs. In particular, the subset of elements
\begin{equation}
  W_{ii;jj} = \left\lbrace
  \begin{array}{l}
 \;1\;  {\rm for} \; j\neq i+s , \\
  \; 0\; {\rm for} \; j=  i+s ,
\end{array}
\right.
\end{equation}
and
\begin{equation}\label{}
  W_{ij;ij} = - \frac{m-1}{d},~\mathrm{for~}i\neq j\,,
\end{equation}
is universal for any set $\mathcal{M}_m$ of $m$ MUBs in $\mathbb{C}^d$. This is a consequence of the property,
\begin{equation}\label{}
  \text{tr}\left[ \sum_{i=0}^{d-1} \left(|i_\alpha \>\<i_\alpha| \otimes |i_\alpha^* \>\<i_\alpha^*|\right) \, \left(|\ell\>\<k| \otimes |\ell\>\<k| \right)\right] = \frac 1d ,
\end{equation}
which holds for any pair $k$ and $\ell$, and a fixed number of $m$ MUBs.

Suppose that $ \mathbf{W}^{\Gamma}(\mathcal{M}_m,s)$ is decomposable. That is, for some $A,B \geq 0$ we have $W = A + B^\Gamma$. Consider three $2 \times 2$ submatrices of $W$ from (\ref{Wijkl}):
\begin{eqnarray}\label{}\nonumber
 A_1&=& \left(\begin{array}{cc} W_{00;ss} & W_{0s;0s} \\ W_{s0;s0} & W_{ss;00} \end{array} \right)
 =  \left( \begin{array}{cc}
           0 & a \\
           a & 1
         \end{array} \right)   \\ \nonumber
     A_2&=&    \left( \begin{array}{cc}
           W_{00;rr} & W_{0r;0r} \\
           W_{r0;r0} & W_{rr;00}
         \end{array} \right) =  \left( \begin{array}{cc}
           1 & a \\
           a & 0
         \end{array} \right)   \\ \nonumber
     A_3&=&    \left( \begin{array}{cc}
           W_{ss;rr} & W_{sr;sr} \\
           W_{rs;rs} & W_{rr;ss}
         \end{array} \right) =  \left( \begin{array}{cc}
           1 & a \\
           a & 0
         \end{array} \right)
\end{eqnarray}
with $a=-(m-1)/d$ and $r=d-s$. Note that $A_1$, $A_2$ and $A_3$ are not positive definite. Therefore, if $ \mathbf{W}^{\Gamma}(\mathcal{M}_m,s)$ is decomposable, then $A_1$, $A_2$ and $A_3$ must originate from $B^\Gamma$, and not from $A$ which is positive definite. We therefore need to check whether $B$ is positive, given that $A_1$, $A_2$ and $A_3$ are submatrices of $B^{\Gamma}$.

If one takes the partial transpose of $B^{\Gamma}$, the off diagonal elements of $A_1$, $A_2$ and $A_3$ that appear in $B^{\Gamma}$ now form a  $3 \times 3$ submatrix of $B$,
\begin{eqnarray}\label{}
 B_{3 \times 3} &=& \left( \begin{array}{ccc}
           W_{00;00} & W_{0s;0s} & W_{0r;0r} \\
           W_{s0;s0} & W_{ss;ss} & W_{sr;sr} \\
           W_{r0;r0} & W_{rs;rs} & W_{rr;rr}
         \end{array} \right)\nonumber\\  &=&  \left( \begin{array}{ccc}
           1 & a & a \\
           a & 1 & a \\
           a & a & 1
         \end{array} \right) .
\end{eqnarray}
It holds that $B_{3 \times 3} \geq 0$ if and only if $a \in [-\frac 12,1]$. Recall that $a = -(m-1)/d$, so that we obtain a necessary condition for the decomposability, i.e., $m \leq \frac d2 + 1$.
This shows that, whenever 
\begin{equation}\label{!}
m > \frac d2 + 1 ,
\end{equation}
the witness $ \mathbf{W}^{\Gamma}(\mathcal{M}_m,s)$ is not decomposable.
\end{proof}
Interestingly, the proof does not rely on the particular choice of $m$ MUBs. Thus, \emph{any} set of mutually unbiased bases (that is sufficiently large) results in a non-decomposable witness able to detect bound entangled states.

\section{Two MUBs give rise to a decomposable witness}\label{sec:2mubs}

We now show that a witness  $\mathbf{W}^{\Gamma}(\mathcal{M}_2,s)$ consisting of two Fourier connected MUBs in $\mathbb{C}^d$ gives rise to a decomposable witness. Suppose that $\mathcal{M}_2$ consists of the standard basis $\mathcal{B}_0=\{\ket{0},\ket{1},\ldots,\ket{d-1}\}$ and its Fourier transform $\mathcal{B}_1=\{\ket{\widetilde 0},\{\ket{\widetilde 1},\ldots,\ket{\widetilde{d-1}}\}$, where
\begin{equation}\label{}
  |\widetilde{k}\> = \frac{1}{\sqrt{d}} \sum_{j=0}^{d-1} \omega^{kj} |j\>\,,
\end{equation}
and $\omega=e^{2\pi i/d}$. Then, we define
\bea\label{}
   \mathbf{W}^{\Gamma}(\mathcal{M}_2,s) = \frac{d+1}{d} \oper_d &\otimes& \oper_d - \sum_{\ell=0}^{d-1} |\ell\>\<\ell| \otimes |\ell+s\>\<\ell +s| \nonumber \\
 &-& \sum_{k=0}^{d-1} |\widetilde{k}\>\<\widetilde{k}| \otimes |\widetilde{k}^*\>\<\widetilde{k}^*|\,.
\eea
One finds
\begin{equation}\label{}
  \sum_{k=0}^{d-1} |\widetilde{k}\>\<\widetilde{k}| \otimes |\widetilde{k}^*\>\<\widetilde{k}^*|  = \frac 1d \sum_{n=0}^{d-1} \sum_{i,j=0}^{d-1} |i\>\<i+n| \otimes |j\>\<j+n|,
\end{equation}
and hence
\begin{equation}\label{}
     \mathbf{W}^{\Gamma}(\mathcal{M}_2,s) = \frac 1d (A(s) + B^\Gamma(s)) ,
\end{equation}
where
\bea\label{}
  A(s) &=& (d-1) (\oper_d \otimes \oper_d - \Pi_s) \nonumber\\
  &-&   \sum_{\substack{n= 0 \\ n\neq s}}^{d-1} \sum_{i\neq j=0}^{d-1} |i\>\<j| \otimes |i+n\>\<j+n| ,
\eea
and $B(s)$ is a sum of $d(d-1)/2$ rank-1 projectors:
\begin{equation}
    B(s) = \sum_{i\neq j=0}^{d-1} |\Psi_{ij}(s)\>\<\Psi_{ij}(s)| ,
\end{equation}
where
\begin{equation}
    |\Psi_{ij}(s)\> =  |i,j+s\> -|j,i+s\>  .
\end{equation}
Evidently $A(s),B(s) \geq 0$, and hence $\mathbf{W}^{\Gamma}(\mathcal{M}_2,s)$ is decomposable for any $s$.

The proof also works for general pairs of Heisenberg-Weyl MUBs in prime dimensions $d=p$, however it is slightly more technical. For example, see the witnesses $\mathbf{W}_{(0,1)}$ and $\mathbf{W}_{(0,3)}$ defined in Sec. \ref{sec:d=3}. One may expect that a decomposition of the witness exists for $\emph{any}$ pair of MUBs, but a proof is hindered by the fact that the classification of pairs is an open problem. Hence, we are unable to describe the general structure of the witness and determine its decomposition. In fact, the classification of pairs of MUBs in $\mathbb{C}^d$ is equivalent to a classification all $d\times d$ complex Hadamard matrices, which is currently known for $d\leq 5$ only \cite{haagerup96}.

The situation is further complicated by the fact that \emph{equivalent} sets of MUBs (up to some unitary transformation) can give rise to different witnesses. For example, a unitary transformation connecting one set of MUBs to another may cause a permutation of the elements of the standard basis. Such a transformation would mean the witness no longer has the particular form given in Eq. (\ref{witness}), and the resulting effect on its (non-)decomposability is unclear.

\section{Example: $d=3$}\label{sec:d=3}

If $d=3$, condition (\ref{!}) implies that three MUBs are sufficient to detect bound entanglement. Actually, we can also show that this condition is necessary. Consider the complete set of four Heisenberg-Weyl MUBs, $\mathcal{B}_0 = \{|0\>,|1\>,|2\>\}$ together with
$\mathcal{B}_1$, $\mathcal{B}_2$ and $\mathcal{B}_3$, which we represent by the following matrices,
\begin{eqnarray}\nonumber
B_1\; =\;  \frac{1}{\sqrt{3}} \left(
\begin{array}{cccc}
1 & 1 & 1 \\
1 & \omega & \omega^* \\
1 & \omega^* & \omega
\end{array}
\right),\,\,
B_2\; =\;  \frac{1}{\sqrt{3}} \left(
\begin{array}{cccc}
1 & 1 & 1 \\
1& \omega & \omega^* \\
\omega^* & \omega & 1
\end{array}
\right),
\\\nonumber
\end{eqnarray}
and $B_3=B_2^*$, where the columns of $B_1$, $B_2$ and $B_3$ correspond to the basis vectors of $\mathcal{B}_1$, $\mathcal{B}_2$ and $\mathcal{B}_3$, respectively. One finds three EWs,
$$\mathbf{W}^{\Gamma}(\{\mathcal{B}_i,\mathcal{B}_j,\mathcal{B}_k\},1):=\mathbf{W}_{(i,j,k)}\,,$$
given by

\begin{equation}\label{}\nonumber
 \mathbf{W}_{(0,1,2)}  =  \left( \begin{array}{ccc|ccc|ccc}
 1 & . & . & . & -\frac{2}{3} & . & . & . & -\frac{2}{3} \\
 .   & . & .    & . & . & \frac{\omega}{3}  & \frac{\omega^*}{3} & . & . \\
 . & . & 1 & \frac{\omega^*}{3} & . & . & . & \frac{\omega}{3}  & . \\  \hline
 . & . & \frac{\omega}{3}  & 1 & . & . & . & \frac{\omega^*}{3}  & . \\
 -\frac{2}{3} & . & . & . & 1 & . & . & . & -\frac{2}{3} \\
 . & \frac{\omega^*}{3} & . & . & . & . & \frac{\omega}{3} & . & . \\   \hline
 . & \frac{\omega}{3}  & . & . & . & \frac{\omega^*}{3}  & . & . & . \\
 . & . & \frac{\omega^*}{3}  & \frac{\omega}{3}  & . & . & . & 1 & . \\
 -\frac{2}{3} & . & . & . & -\frac{2}{3} & . & . & . & 1 \\
\end{array}
\right),
\end{equation}
  \begin{equation}\nonumber
   \mathbf{W}_{(0,1,3)} = \mathbf{W}^*_{(0,1,2)}\,,
   \end{equation}
and
\begin{equation}\label{}\nonumber
 \mathbf{W}_{(0,2,3)} =  \left( \begin{array}{ccc|ccc|ccc}
 1 & . & . & . & -\frac{2}{3} & . & . & . & -\frac{2}{3} \\
 .   & . & .    & . & . & \frac{1}{3}  & \frac{1}{3} & . & . \\
 . & . & 1 & \frac{1}{3} & . & . & . & \frac{1}{3}  & . \\  \hline
 . & . & \frac{1}{3}  & 1 & . & . & . & \frac{1}{3}  & . \\
 -\frac{2}{3} & . & . & . & 1 & . & . & . & -\frac{2}{3} \\
 . & \frac{1}{3} & . & . & . & . & \frac{1}{3} & . & . \\   \hline
 . & \frac{1}{3}  & . & . & . & \frac{1}{3}  & . & . & . \\
 . & . & \frac{1}{3}  & \frac{1}{3}  & . & . & . & 1 & . \\
 -\frac{2}{3} & . & . & . & -\frac{2}{3} & . & . & . & 1 \\
\end{array}
\right) ,
\end{equation}
where the dots represent zeros.

It is, therefore, clear that the matrix structure of $\mathbf{W}^{\Gamma}(\mathcal{M}_3,1)$ depends on a particular set of $\mathcal{M}_3$. However, all three witnesses are non-decomposable and they detect the same PPT entangled state,

\begin{equation}\label{}
 \rho_x = \frac{1}{\mathcal{N}} \left(\begin{array}{ccc|ccc|ccc}
 1 & . & . & . & 1 & . & . & . & 1 \\
 . & 1/x & . & . & . & . & . & . & . \\
 . & . & x & . & . & . & . & . & . \\  \hline
 . & . & . & x & . & . & . & . & . \\
 1 & . & . & . & 1 & . & . & . & 1 \\
 . & . & . & . & . & 1/x & . & . & . \\   \hline
 . & . & . & . & . & . & 1/x & . & . \\
 . & . & . & . & . & . & . & x & . \\
 1 & . & . & . & 1 & . & . & . & 1 \\
\end{array} \right) ,
\end{equation}
with $x \neq 1$ (for $x=1$ it is separable).

Similarly, one finds three witnesses $\mathbf{W}^{\Gamma}(\{\mathcal{B}_i,\mathcal{B}_j\},1):=\mathbf{W}_{(i,j)}$ as follows:
\begin{equation}\label{}\nonumber
 \mathbf{W}_{(0,1)} =  \left( \begin{array}{ccc|ccc|ccc}
 1 & . & . & . & -\frac{1}{3} & . & . & . & -\frac{1}{3} \\
 .   & . & .    & . & . & -\frac{1}{3}  & -\frac{1}{3} & . & . \\
 . & . & 1 & -\frac{1}{3} & . & . & . & -\frac{1}{3}  & . \\  \hline
 . & . & -\frac{1}{3}  & 1 & . & . & . & -\frac{1}{3}  & . \\
 -\frac{1}{3} & . & . & . & 1 & . & . & . & -\frac{1}{3} \\
 . & -\frac{1}{3} & . & . & . & . & -\frac{1}{3} & . & . \\   \hline
 . & -\frac{1}{3}  & . & . & . & -\frac{1}{3}  & . & . & . \\
 . & . & -\frac{1}{3}  & -\frac{1}{3}  & . & . & . & 1 & . \\
 -\frac{1}{3} & . & . & . & -\frac{1}{3} & . & . & . & 1 \\
\end{array}
\right)
\end{equation}

\begin{equation}\label{}\nonumber
\mathbf{W}_{(0,3)}  =  -\left( \begin{array}{ccc|ccc|ccc}
 -1 & . & . & . & \frac{1}{3} & . & . & . & \frac{1}{3} \\
 .   & . & .    & . & . & \frac{\omega}{3}  & \frac{\omega^*}{3} & . & . \\
 . & . & -1 & \frac{\omega^*}{3} & . & . & . & \frac{\omega}{3}  & . \\  \hline
 . & . & \frac{\omega}{3}  & -1 & . & . & . & \frac{\omega^*}{3}  & . \\
 \frac{1}{3} & . & . & . & -1 & . & . & . & \frac{1}{3} \\
 . & \frac{\omega^*}{3} & . & . & . & . & \frac{\omega}{3} & . & . \\   \hline
 . & \frac{\omega}{3}  & . & . & . & \frac{\omega^*}{3}  & . & . & . \\
 . & . & \frac{\omega^*}{3}  & \frac{\omega}{3}  & . & . & . & -1 & . \\
 \frac{1}{3} & . & . & . & \frac{1}{3} & . & . & . & -1 \\
\end{array}
\right)
\end{equation}
and
\begin{equation}\label{}\nonumber
   \mathbf{W}_{(0,2)}  = \mathbf{W}^*_{(0,3)}  .
\end{equation}

Again, the matrix structure of $\mathbf{W}(\mathcal{M}_2,1)$ depends on the particular set of $\mathcal{M}_2$. However, all three witnesses are decomposable. In particular   $\mathbf{W}_{(0,1)} = A + B^\Gamma$, where

\begin{equation}\label{}
A= \frac 13 \left(
\begin{array}{ccc|ccc|ccc}
 2 & . & . & . & -1 & . & . & . & -1 \\
 . & . & . & . & . & . & . & . & . \\
 . & . & 2 & -1 & . & . & . & -1 & . \\  \hline
 . & . & -1 & 2 & . & . & . & -1 & . \\
 -1 & . & . & . & 2 & . & . & . & - 1 \\
 . & . & . & . & . & . & . & . & . \\  \hline
 . & . & . & . & . & . & . & . & . \\
 . & . & -1 & -1 & . & . & . & 2 & . \\
 -1 & . & . & . & -1 & . & . & . & 2 \\
\end{array}
\right)
\end{equation}
and
\begin{equation}\label{}
B = \frac 13
\left(
\begin{array}{ccc|ccc|ccc}
 1 & . & . & . & . & . & . & -1 & . \\
 . & . & . & . & . & . & . & . & . \\
 . & . & 1 & . & -1 & . & . & . & . \\  \hline
 . & . & . & 1 & . & . & . & . & -1 \\
 . & . & -1 & . & 1 & . & . & . & . \\
 . & . & . & . & . & . & . & . & . \\  \hline
 . & . & . & . & . & . & . & . & . \\
 -1 & . & . & . & . & . & . & 1 & . \\
 . & . & . & -1 & . & . & . & . & 1 \\
\end{array}
\right) ,
\end{equation}
with $A,B\geq 0$. Similar decompositions hold for the remaining two witnesses.

\section{Special dimensions: $d=2^r$}\label{sec:special}

We now improve Theorem \ref{mainresult} by showing that 3 MUBs (i.e., $m=\tfrac{d}{2} +1$) are sufficient for bound entanglement detection when $d=4$. In particular, we find specific examples of three MUBs which detect bound entanglement. It remains open whether all triples (for which there is an infinite family \cite{brierley10}) are sufficient for bound entanglement detection.

Consider the following set of three MUBs in $\mathbb{C}^4$, denoted by $\{\mathcal{B}_0, \mathcal{B}_1,\mathcal{B}_x\}$, where $\mathcal{B}_0$ is the canonical basis, $\mathcal{B}_1$, in matrix form, is given by
\begin{eqnarray}\label{}
  B_1 = \frac 12 \left( \begin{array}{cccc} 1 & 1 & 1 & 1 \\ 1 & 1 & -1 & -1 \\ 1 & -1 & -1 & 1 \\ 1 & -1 & 1 & -1 \end{array} \right) \ ,
  \end{eqnarray}
and $\mathcal{B}_x$, a third mutually unbiased basis, is chosen as either
  \begin{eqnarray}\label{}
  B_{\text{ext}} = \frac 12 \left( \begin{array}{cccc} 1 & 1 & 1 & 1 \\ i & -i & i & -i \\ -1 & -1 & 1 & 1\\ i & -i & -i & i  \end{array} \right) \ ,
    \end{eqnarray}
    or
    \begin{eqnarray}\label{}
   B_{\text{unext}} = \frac 12 \left( \begin{array}{cccc} 1 & 1 & 1 & 1 \\ 1 & 1 & -1 & -1 \\ -1 & 1 & 1 & -1 \\ 1 & -1 & 1 & -1 \end{array} \right) \ .
\end{eqnarray}
The three MUBs $\{\mathcal{B}_0, \mathcal{B}_1,\mathcal{B}_{\text{ext}}\}$ form an \emph{extendible} set, which is a subset of the complete set of five Heisenberg-Weyl MUBs. In contrast, $\{\mathcal{B}_0, \mathcal{B}_1,\mathcal{B}_{\text{unext}}\}$ is \emph{unextendible}, in the sense that no other basis is mutually unbiased to all three bases \cite{brierley10}.
We label the associated EWs
\begin{equation}
\mathbf{W}^\Gamma(\{\mathcal{B}_0, \mathcal{B}_1,\mathcal{B}_x\},1):=\mathbf{W}_{x}\,.
\end{equation}
Now we have the following result.

\begin{Prop}\label{prop:d=4} The entanglement witnesses $$\mathbf{W}_{\emph{ext}} \,\,\,\,\text{and} \,\,\,\,\mathbf{W}_{\emph{unext}}$$ are non-decomposable.
\end{Prop}

\begin{proof}
In the Appendix (Eqs. (\ref{best=1}) and (\ref{best=2})) we present the matrix form of these witnesses. Non-decomposability follows by constructing families of PPT states $\rho_{a}$ and $\rho_b$, defined in (\ref{rho-a}) and (\ref{rho-b}), which violate the separability criterion. Firstly, we see that
\begin{equation}\nonumber
  \tr[ \mathbf{W}_{\text{ext}}\rho_{a}] = \frac{4}{\mathcal{N}}(a-1) < 0 ,
\end{equation}
for $a<1$. Since $\rho_{a}$ is PPT in this region and the separability condition is clearly violated when $a<1$, this proves that $\mathbf{W}_{\text{ext}}$ is non-decomposable.

Similarly, after a simple calculation, we have
\begin{equation}\label{}\nonumber
  \tr[\mathbf{W}_{\text{unext}}\rho_b] =\frac{4}{\mathcal{N}} (b -1) < 0 ,
\end{equation}
for $b <1$. This proves that $\mathbf{W}_{\text{unext}}$ is indeed non-decomposable (and $\rho_b$ is PPT entangled for $b<1$).

Hence, for $d=4$ we have shown that both extendible and unextendible MUBs give rise to a non-decomposable EWs. As a by-product we present two new families (\ref{rho-a}) and (\ref{rho-b}) of PPT entangled states in a $4\times 4$ dimensional system that are not locally unitarily equivalent.
\end{proof}

Furthermore, from Sec. \ref{sec:2mubs}, the entanglement witness $\mathbf{W}^{\Gamma}(\mathcal{M}_2,1)$, with a Fourier connected pair of measurements, is decomposable. Hence, while we have not ruled out all possible pairs from detecting bound entanglement, the extension to a third basis appears necessary. When $d=2^r$, we predict that $\frac d 2+1$ MUBs are both necessary and sufficient for the witness $\mathbf{W}^{\Gamma}(\mathcal{M}_m,s)$ to be non-decomposable.

\begin{conjecture}\label{con} If $d=2^r$ and $r>1$, the minimal number of MUBs required for bound entanglement detection equals $m=\frac d 2+1=2^{r-1}+1$.
\end{conjecture}

We point out that the three MUBs used to construct $\mathbf{W}_{\text{unext}}$ have the unusual feature of being bases in $\mathbb{R}^4$. It is known that the number of \emph{real} MUBs in $\mathbb{R}^d$ is at most $\frac{d}{2}+1$, and the bound is tight when $d=4^r$ \cite{boykin05}. This suggests the possibility of finding non-decomposable witnesses from real MUBs in these special dimensions. It would be interesting to determine the usefulness of real MUBs for entanglement detection \cite{add1}, especially in light of their various constructions, e.g. \cite{boykin05}.

The above conjecture, together with the example of the witness $\mathbf{W}_{\text{unext}}$, suggests that \emph{large} sets of unextendible MUBs in higher dimensions can be used for bound entanglement detection. For $d=4,8$ it is known that $m=d/2+1$ unextendible MUBs exist, and it is conjectured that such bases also occur when $d=2^r$ \cite{mandayam14}. In contrast, smaller sets of unextendible MUBs are not expected to yield examples of non-decomposable witnesses, including those constructed up to $d=16$ in \cite{grassl17}, and the $p+1$ unextendible MUBs found in \cite{hedge15} when $d=p^2$.


\section{Conclusions}

In this article we investigated the minimal number of mutually unbiased bases needed to detect bound entangled states. We have shown that sets of $m$ MUBs in $\mathbb{C}^d$, when $m>d/2+1$, give rise to non-decomposable witnesses. Thus, bound entanglement verification can be achieved with significantly fewer measurements than a complete set of $d+1$ MUBs. This presents a more favourable, experimentally friendly, detection approach compared with typical quantum state tomography methods, or witnesses composed of complete sets. It also opens up the possibility of applying these witnesses to detect PPT entangled states in dimensions \emph{without} the existence of a complete set, i.e. in non-prime-power dimensions.

We have shown that the original EW composed of $d+1$ MUBs in \cite{PhysRevA.86.022311} is decomposable, and therefore cannot detect bound entangled states. We also found that pairs of Fourier connected MUBs (as well as Heisenberg-Weyl pairs in prime dimensions) are unable to certify bound entanglement. When $d=3$, we proved that three MUBs are necessary and sufficient to detect bound entangled states. In $d=4$, we strengthened Theorem \ref{mainresult} to show that PPT entangled states can be detected with $m=3=\frac d2+1$ MUBs. As a consequence, we found new families of bipartite bound entangled states in a ($4\times4$) system, and conjectured $m=\frac d2+1$ MUBs are required to detect bound entanglement if $d=2^r$ with $r>1$.

In recent work it has been shown that certain sets of MUBs are more useful than others---a consequence of the unitary inequivalence of different classes of MUBs---in information processing tasks \cite{designolle19,aguilar,hiesmayr20}. In contrast, Theorem \ref{mainresult} shows that since \emph{any} set of $m$ MUBs ($m>d/2+1$) detects bound entangled states, the non-decomposability of a witness is a universal property of mutually unbiased bases. However, since different equivalence classes give rise to different EWs, an interesting interplay is revealed between the universality of MUBs for non-decomposable witnesses and the more local dependence on the choice of MUBs in detecting a given bound entangled state. The nature of such connections would be worth investigating further since they may shed light on the structure of the set of bound entangled states.

Other interesting directions include studying the optimality of these witnesses, and determining whether a reduction in the number of measurements can result in a decomposable witness transitioning to a non-decomposable one. Further open questions include whether \emph{every} pair of MUBs is unable to detect bound entanglement, and, more generally, the minimal number of MUBs needed to detect bound entangled states. Another avenue to explore is the construction of new families of bipartite PPT entangles states, which are detected by the large class of non-decomposable EWs we have studied.



\vspace{1cm}
\section*{Acknowledgements}
JB is supported by the National Research Foundation of Korea (Grant No. NRF-2021R1A2C200 6309, NRF-2022M1A3C2069728) and the Institute for Information \& Communication Technology Promotion (IITP) (the ITRC Program/IITP-2022-2018-0-01402). AB and DC were supported by the Polish National Science Centre project 2018/30/A/ST2/00837. BCH acknowledges gratefully that this research was funded in whole, or in part, by the  Austrian Science Fund (FWF) project P36102-N. DM has received funding from the European Union's Horizon 2020 research and innovation programme under the Marie Sk\l{}odowska-Curie grant agreement No. 663830. DM also acknowledges financial support by the TEAM-NET project co-financed by the EU within the Smart Growth Operational Programme (Contract No. POIR.04.04.00-00-17C1/18-00).

\onecolumngrid

\section*{Appendix: Non-decomposable witness when $d=4$}

As shown in Prop. \ref{prop:d=4}, we find two non-decomposable EWs using sets of three MUBs. The witness corresponding to the three Heisenberg-Weyl MUBs is given by
\begin{equation}\label{best=1}
\mathbf{W}_{\text{ext}} = \left(
\begin{array}{cccc|cccc|cccc|cccc}
 1 & . & . & . & . & -\frac{1}{2} & . & . & . & . & -\frac{1}{2} & . & . & . & . & -\frac{1}{2} \\
 . & . & . & . & . & . & . & . & . & . & . & . & . & . & -\frac{1}{2} & . \\
 . & . & 1 & . & . & . & . & . & -\frac{1}{2} & . & . & . & . & . & . & . \\
 . & . & . & 1 & . & . & -\frac{1}{2} & . & . & . & . & . & . & . & . & . \\  \hline
 . & . & . & . & 1 & . & . & . & . & . & . & -\frac{1}{2} & . & . & . & . \\
 -\frac{1}{2} & . & . & . & . & 1 & . & . & . & . & -\frac{1}{2} & . & . & . & . & -\frac{1}{2} \\
 . & . & . & -\frac{1}{2} & . & . & . & . & . & . & . & . & . & . & . & . \\
 . & . & . & . & . & . & . & 1 & . & . & . & . & . & -\frac{1}{2} & . & . \\ \hline
 . & . & -\frac{1}{2} & . & . & . & . & . & 1 & . & . & . & . & . & . & . \\
 . & . & . & . & . & . & . & . & . & 1 & . & . & -\frac{1}{2} & . & . & . \\
 -\frac{1}{2} & . & . & . & . & -\frac{1}{2} & . & . & . & . & 1 & . & . & . & . & -\frac{1}{2} \\
 . & . & . & . & -\frac{1}{2} & . & . & . & . & . & . & . & . & . & . & . \\  \hline
 . & . & . & . & . & . & . & . & . & -\frac{1}{2} & . & . & . & . & . & . \\
 . & . & . & . & . & . & . & -\frac{1}{2} & . & . & . & . & . & 1 & . & . \\
 . & -\frac{1}{2} & . & . & . & . & . & . & . & . & . & . & . & . & 1 & . \\
 -\frac{1}{2} & . & . & . & . & -\frac{1}{2} & . & . & . & . & -\frac{1}{2} & . & . & . & . & 1 \\
\end{array}
\right)\,.
\end{equation}
We consider the family of PPT states
\begin{equation}\label{rho-a}
\rho_{a}=\frac{1}{\mathcal{N}}\left(
\begin{array}{cccc|cccc|cccc|cccc}
 1 & . & . & . & . & 1 & . & . & . & . & 1 & . & . & . & . & 1 \\
 . & 1/a & . & . & . & . & . & . & . & . & . & . & . & . & 1 & . \\
 . & . & 1 & . & . & . & . & . & 1 & . & . & . & . & . & . & . \\
 . & . & . & a & . & . & 1 & . & . & . & . & . & . & . & . & . \\ \hline
 . & . & . & . & a & . & . & . & . & . & . & 1 & . & . & . & . \\
 1 & . & . & . & . & 1 & . & . & . & . & 1 & . & . & . & . & 1 \\
 . & . & . & 1 & . & . & 1/a & . & . & . & . & . & . & . & . & . \\
 . & . & . & . & . & . & . & 1 & . & . & . & . & . & 1 & . & . \\ \hline
 . & . & 1 & . & . & . & . & . & 1 & . & . & . & . & . & . & . \\
 . & . & . & . & . & . & . & . & . & a & . & . & 1 & . & . & . \\
 1 & . & . & . & . & 1 & . & . & . & . & 1 & . & . & . & . & 1 \\
 . & . & . & . & 1 & . & . & . & . & . & . & 1/a & . & . & . & . \\ \hline
 . & . & . & . & . & . & . & . & . & 1 & . & . & 1/a & . & . & . \\
 . & . & . & . & . & . & . & 1 & . & . & . & . & . & 1 & . & . \\
 . & 1 & . & . & . & . & . & . & . & . & . & . & . & . & a & . \\
 1 & . & . & . & . & 1 & . & . & . & . & 1 & . & . & . & . & 1 \\
\end{array}
\right)\,,
\end{equation}
with $a>0$, which are verified as entangled by $\mathbf{W}_{\text{ext}}$ when $a<1$.

The second witness, constructed from a set of three unextendible MUBs, takes the form
\begin{equation}\label{best=2}
\mathbf{W}_{\text{unext}} = \left(
\begin{array}{cccc|cccc|cccc|cccc}
 1 & . & . & . & . & -\frac{1}{2} & . & . & . & . & -\frac{1}{2} & . & . & . & . & -\frac{1}{2} \\
 . & . & . & . & -\frac{1}{2} & . & . & . & . & . & . & . & . & . & . & . \\
 . & . & 1 & . & . & . & . & . & -\frac{1}{2} & . & . & . & . & . & . & . \\
 . & . & . & 1 & . & . & . & . & . & . & . & . & -\frac{1}{2} & . & . & . \\ \hline
 . & -\frac{1}{2} & . & . & 1 & . & . & . & . & . & . & . & . & . & . & . \\
 -\frac{1}{2} & . & . & . & . & 1 & . & . & . & . & -\frac{1}{2} & . & . & . & . & -\frac{1}{2} \\
 . & . & . & . & . & . & . & . & . & -\frac{1}{2} & . & . & . & . & . & . \\
 . & . & . & . & . & . & . & 1 & . & . & . & . & . & -\frac{1}{2} & . & . \\ \hline
 . & . & -\frac{1}{2} & . & . & . & . & . & 1 & . & . & . & . & . & . & . \\
 . & . & . & . & . & . & -\frac{1}{2} & . & . & 1 & . & . & . & . & . & . \\
 -\frac{1}{2} & . & . & . & . & -\frac{1}{2} & . & . & . & . & 1 & . & . & . & . & -\frac{1}{2} \\
 . & . & . & . & . & . & . & . & . & . & . & . & . & . & -\frac{1}{2} & . \\ \hline
 . & . & . & -\frac{1}{2} & . & . & . & . & . & . & . & . & . & . & . & . \\
 . & . & . & . & . & . & . & -\frac{1}{2} & . & . & . & . & . & 1 & . & . \\
 . & . & . & . & . & . & . & . & . & . & . & -\frac{1}{2} & . & . & 1 & . \\
 -\frac{1}{2} & . & . & . & . & -\frac{1}{2} & . & . & . & . & -\frac{1}{2} & . & . & . & . & 1 \\
\end{array}
\right)\,.
\end{equation}
Note, that $\mathbf{W}_{\text{unext}}^\Gamma=\mathbf{W}_{\text{unext}} $. This is shown to be non-decomposable by considering the family of PPT states,
\begin{equation}\label{rho-b}
\rho_b = \frac{1}{\mathcal{N}}\left(
\begin{array}{cccc|cccc|cccc|cccc}
 1 & . & . & . & . & 1 & . & . & . & . & 1 & . & . & . & . & 1 \\
 . & 1/b & . & . & 1 & . & . & . & . & . & . & . & . & . & . & . \\
 . & . & 1 & . & . & . & . & . & 1 & . & . & . & . & . & . & . \\
 . & . & . & b & . & . & . & . & . & . & . & . & 1 & . & . & . \\ \hline
 . & 1 & . & . & b & . & . & . & . & . & . & . & . & . & . & . \\
 1 & . & . & . & . & 1 & . & . & . & . & 1 & . & . & . & . & 1 \\
 . & . & . & . & . & . & 1/b & . & . & 1 & . & . & . & . & . & . \\
 . & . & . & . & . & . & . & 1 & . & . & . & . & . & 1 & . & . \\ \hline
 . & . & 1 & . & . & . & . & . & 1 & . & . & . & . & . & . & . \\
 . & . & . & . & . & . & 1 & . & . & b & . & . & . & . & . & . \\
 1 & . & . & . & . & 1 & . & . & . & . & 1 & . & . & . & . & 1 \\
 . & . & . & . & . & . & . & . & . & . & . & 1/b & . & . & 1 & . \\ \hline
 . & . & . & 1 & . & . & . & . & . & . & . & . & 1/b & . & . & . \\
 . & . & . & . & . & . & . & 1 & . & . & . & . & . & 1 & . & . \\
 . & . & . & . & . & . & . & . & . & . & . & 1 & . & . & b & . \\
 1 & . & . & . & . & 1 & . & . & . & . & 1 & . & . & . & . & 1 \\
\end{array}
\right)\,,
\end{equation}
with $b>0$, which are verified as entangled when $b<1$.

\end{document}